\documentclass[11pt]{amsart}

\theoremstyle{definition}
\newtheorem{theorem}{Theorem}
\newtheorem{corollary}[theorem]{Corollary}
\newtheorem{proposition}[theorem]{Proposition}
\newtheorem{lemma}[theorem]{Lemma}
\newtheorem{definition}[theorem]{Definition}
\newtheorem{example}[theorem]{Example}
\newtheorem{notation}[theorem]{Notation}
\newtheorem{remark}[theorem]{Remark}

\newcommand{\numberset}{\mathbb}
\newcommand{\N}{\numberset{N}}

\newcommand{\F}{\numberset{F}}

\newcommand{\mC}{\mathcal{C}}
\newcommand{\mP}{\mathcal{P}}
\newcommand{\mG}{\mathcal{G}}
\newcommand{\mH}{\mathcal{H}}
\usepackage{hyperref}
\newcommand{\mA}{\mathcal{A}}
\newcommand{\mN}{\mathcal{N}}
\usepackage[margin=3cm]{geometry}

\usepackage{fouriernc}

\usepackage{amsaddr}

\author{Elisa Gorla$^1$ and  Alberto Ravagnani$^2$}
\address{Institut de Math\'{e}matiques
, Universit\'{e} de Neuch\^{a}tel\\Rue
Emile-Argand 11, CH-2000 Neuch\^{a}tel, Switzerland}
\email{$^1$elisa.gorla@unine.ch}
\email{$^2$alberto.ravagnani@unine.ch}

\usepackage{algorithm}
\usepackage{algorithmic}

\begin{document}

\title{Partial Spreads in Random Network Coding}

\subjclass[2010]{11T71}
\keywords{network code, spread code, subspace distance}

\begin{abstract}
 Following the approach by R. K\"{o}tter and F. R. Kschischang,
we study network codes as families of $k$-dimensional
linear subspaces of a vector space $\F_q^n$, $q$ being a prime power and $\F_q$
the finite field with $q$ elements.
 In particular, following an idea in finite projective geometry, we
introduce a class of network codes which we call \textit{partial
  spread codes}. Partial spread codes naturally generalize spread
codes. In this paper we provide an easy description of such codes in
terms of matrices, discuss their maximality, and provide an efficient
decoding algorithm. 
\end{abstract}

\maketitle

\setcounter{section}{-1}

\section{Introduction}\label{intr0}

The topology of a network is well-modeled by a directed multigraph. Vertices
without incoming edges play the role of \textit{sources} and vertices without
outgoing edges play the role of \textit{sinks}. Vertices which are neither
sources nor sinks are called \textit{nodes}. The interest in network modeling
 is due to its several applications in technology (distributed
storage, peer-to-peer networking and, in particular, wireless communications).

In \cite{origine} Ahlswede, Cai, Li, and Yeung discovered that the
information rate may be improved by employing coding at the nodes of a
network (instead of simply routing). Moreover, Li, Cai and Yeung
proved in \cite{origine2} that, in a multicasting situation, maximal
information rate can be achieved by allowing the nodes to transmit linear 
combinations of the inputs they receive, provided that the size of the
base field is large enough.

A turning point in the study of linear network codes was the paper
\cite{KK1} by R. K\"{o}tter and F. R. Kschischang. The authors
suggested an algebraic approach to the topic, developing a clear and
rigorous mathematical setup. Interesting connections with classical
projective geometry also emerged. Several other interesting papers
followed the same approach, e.g., \cite{ref1}, \cite{ref2}, and
\cite{KK2}.

In this paper, we propose and study a class of network codes,
which fit within the same framework. In Section \ref{intr} the algebraic
approach by K\"{o}tter and Kschischang is briefly recalled. In Section
\ref{sec2} we introduce a family of network codes which we call
partial spread codes, and which generalize spread codes (see
\cite{GR}). Our codes have the same cardinality and distance
distribution as the codes proposed in~\cite{EV}. The elements of our
codes however are given as rowspaces of appropriate matrices in block
form. The structure of this family of matrices allow us to derive
properties of the code, which we discuss in Section \ref{sec3}. In
particular, we establish the maximality of partial spread codes with
respect to containment.
Based on the same block matrix structure, in Section \ref{sec4} we
are able to give an efficient decoding algorithm.
% based on a decoding procedure for \textit{Reed-Solomon like} codes
% (\cite{KK1}, Section V.B).

\section{Preliminaries}\label{intr}

Let $q$ be a prime power and let $\F_q$ denote the finite field with $q$
elements.
Fix an integer $n>1$ and let $\mathcal{P}(\F_q^n)$ be the \textbf{projective
geometry}
of $\F_q^n$, i.e., the set of all the vector subspaces of $\F_q^n$.
 Following \cite{KK1}, a $q$-ary \textbf{network code} of length $n$ is defined
to be a
subset $\mC \subseteq \mathcal{P}(\F_q^n)$ with at least two elements.
The \textbf{subspace distance} on $\mP(\F_q^n)$ is the distance map
 $d:\mP(\F_q^n) \times \mP(\F_q^n) \to \N$ defined, for any $U,V \in
\mP(\F_q^n)$, by
$$d(U,V):=\dim(U)+\dim(V)-2\dim(U \cap V).$$ As in classical Coding Theory, 
the \textbf{minimum distance} of a network-code
$\mC \subseteq \mP(\F_q^n)$ is the integer
$d(\mC):= \min \{ d(U,V) : U,V \in \mC, U \neq V \}$. 
The \textbf{maximum dimension} of $\mC$ is denoted and
defined by $\ell(\mC):=\max_{V \in \mC} \ \dim(V)$.
Let us briefly recall from \cite{KK1} the framework for errors and erasures
in random network coding.
If $1 \le e <n$ is an integer, then an $e$-\textbf{erasure} on an element 
$V \in \mP(\F_q^n)$ such that $\dim(V) \ge e$ 
is the projection of $V$ onto
an $e$-dimensional subspace of $V$. In other words, an $e$-erasure replaces $V$
with
an $e$-dimensional subspace of $V$. A $t$-dimensional \textbf{error} $E$ on an
element
$V \in \mP(\F_q^n)$ corresponds to the direct sum $V \oplus E$, where $E \in
\mP(\F_q^n)$,
$\dim(E)=t$ and $E \cap V=\{ 0\}$. If $\mC \subseteq \mP(\F_q^n)$ is a network
code, then an input codeword
$V \in \mC$ and its output $U \in \mP(\F_q^n)$ are related by $U=\mH_e(V)
\oplus
E$, 
where $1 \le e \le \dim(V)$, $\mH_e$ is an $e$-erasure operator and $E \in
\mP(\F_q^n)$ the error. As usual, one can bound the number of erasures and
errors
that can take place such that a minimum distance decoder is guaranteed to
successfully
return the sent codeword. 

\begin{theorem}[\cite{KK1}, Theorem 2]
 Let $\mC \subseteq \mP(\F_q^n)$ be a network-code of minimum distance $d$. 
Assume that an input $V \in \mC$ and its output
$U \in \mP(\F_q^n)$ are related by $U=\mH_e(V) \oplus E$,
where $e \le \ell(\mC)$, $\mH_e$ is an $e$-erasure and $E\in \mP(\F_q^n)$ is an
error.
 Set $t:=\dim(E)$. A minimum distance decoder corrects $U$ in $V$, 
provided that $2(t+\ell(\mC)-e)<d$.
\end{theorem}

A natural class of network codes is obtained by considering subsets of
$\mP(\F_q^n)$, all of whose elements have the same dimension $1 \le k
\le n-1$. Such codes are called \textbf{constant dimension}
codes. By introducing the Grassmannian variety
$$\mG_q(k,n):= \{ V \in \mP(\F_q^n) : \dim(V)=k \},$$
a $q$-ary constant dimension network code of lenght $n$ and dimension $k$ is
simply a subset
$\mC \subseteq \mG_q(k,n)$ of at least two elements. It easily follows from the
definition that
any constant dimension
network code has
even minimum distance.

\begin{remark}
 The cardinality of the Grassmannian variety $\mG_q(k,n)$ is
known to be 
$${\begin{bmatrix} n \\ k \end{bmatrix}}_q := \frac{(q^n-1)(q^{n-1}-1) \cdots 
(q^{n-k+1}-1)}{(q^k-1)(q^{k-1}-1) \cdots (q-1)}= \prod_{i=0}^{k-1}
\frac{q^{n-i}-1}{q^{k-i}-1}.$$
\end{remark}

Bounds on the size of network codes have been discussed in depth by R. K\"{o}tter
and F. R. Kschischang in \cite{KK1}. More recently, in~\cite{EV},
T. Etzion and A. Vardy obtained other bounds.

\begin{theorem}[Singleton-like Bound, \cite{KK1}, Theorem 9]\label{sing}
 Let $\mC \subseteq \mG_q(k,n)$ be a network code of minimum distance $d$. Then
$$|\mC| \le {\begin{bmatrix} n-(d-2)/2 \\ \max\{ k, n-k\} \end{bmatrix}}_q.$$
\end{theorem}

The family of \textit{spread codes} has been introduced in \cite{GR},
and an efficient decoding algorithm for such codes has been provided
in \cite{GR2}.

\begin{definition}\label{prima}
 A $k$-\textbf{spread} of $\F_q^n$ is a collection of subspaces
${\{V_i\}}_{i=1}^t$ of $\F_q^n$ (here we take $k < n$) such that
\begin{enumerate}
\item $\dim V_i=\dim V_j=k$ for any $i,j \in \{ 1,...,t\}$,
\item $V_i \cap V_j=\{ 0\}$ whenever $i \neq j$,
\item $\F_q^n=\bigcup_{i=1}^t V_i$.
\end{enumerate}
\end{definition}

\begin{remark}\label{sppp}
 A $k$-spread of $\F_q^n$ exists if and only if $k$ divides $n$ (see \cite{H},
Corollary 4.17). From the definition we see that if ${\{V_i\}}_{i=1}^t$ is
a $k$-spread of $\F_q^n$ then
$t=(q^n-1)/(q^k-1)$. Being a subset of the Grassmannian $\mG_q(k,n)$, a
$k$-spread in $\F_q^n$ is a 
$q$-ary network
code of lenght $n$, dimension $k$ and minimum distance $2k$. It is easily
checked
that spread codes
meet the Singleton-like bound (Theorem \ref{sing}).
\end{remark}

\section{Partial spread codes}\label{sec2}

In this section we introduce a generalization of the definition of
spread and a related family of network codes, whose parameters $k$ and $n$ 
can be chosen freely.

\begin{definition}
 A \textbf{partial $k$-spread} of $\F_q^n$ is a subset $\mC \subseteq
\mG_q(k,n)$
such that $U \cap V= \{ 0\}$ for any $U,V \in \mC$ with $U \neq V$. A
 partial
$k$-spread of $\F_q^n$ with at least two elements is  a 
$q$-ary network code of lenght $n$, dimension $k$ and minimum distance $2k$. We
will call
such a code a \textbf{partial spread code}.
\end{definition}

\begin{lemma}\label{bound}
 Let $\mC \subseteq \mG_q(k,n)$ be a partial spread code. Denote by $r$ the
remainder
obtained dividing $n$ by $k$. Then
$$|C| \le \frac{q^n-q^r}{q^k-1}.$$
\end{lemma}

\begin{proof}
Since $\mC$ is a set of $k$-dimensional vector subspaces of $\F_q^n$ with 
trivial 
pairwise intersections, we deduce
$|\mC|\cdot (q^k-1)+1 \le q^n$. Since $k$ divides $n-r$, $(q^{n-r}-1)/(q^k-1)$
is an integer. Hence
$$ |\mC| \le \Bigg \lfloor \frac{q^n-1}{q^k-1} \Bigg \rfloor =\Bigg \lfloor
\frac{q^r(q^{n-r}-1)}{q^k-1} +  \frac{q^r-1}{q^k-1} \Bigg \rfloor= 
\frac{q^n-q^r}{q^k-1}.$$ 
\end{proof}

The bound given in Lemma \ref{bound} admits some non-trivial improvements. See
\cite{Beu_2} and
 \cite{Drake_etc} for details. The following lower bound for partial
$k$-spread in 
$\F_q^n$ 
 is due to A. Beutelspacher (see \cite{Beu_1} for a non-constructive proof). 

\begin{lemma}\label{low}
 Let $q$ be a prime power and let $1 \le k < n$ be integers. Write
$n=hk+r$ with $0 \le r \le k-1$. Denote by $\mA_q(k,n,2k)$ the largest possible
size of a network code $\mC \subseteq \mG_q(k,n)$ of minimum distance $2k$. Then
$$\mA_q(n,k,2k) \ge (q^n-q^r)/(q^k-1)-q^r+1.$$
\end{lemma}

\begin{remark}
 An alternative proof of Lemma \ref{low} is given in \cite{EV}, Theorem 11. For
interesting
discussions on the sharpness of the bound see \cite{FinGeo1} and \cite{FinGeo2}.
\end{remark}

  Here we introduce a construction for
partial spread codes whose size attains the lower bound of Lemma \ref{low}.
Notice that
the vector spaces of the partial spread are given as row spaces of appropriate
easy-computable
matrices.

\begin{lemma}[\cite{LN}, Ch. 2.5]\label{inver}
 Let $q$ be a prime power and let $\F_q$ be the finite field with $q$ elements.
Choose an irreducible monic polynomial $p \in \F_q[x]$ of degree $k \ge 1$ and
write
$p=\sum_{i=0}^k p_ix^i$.  Define the \textbf{companion matrix} of $p$ 
by
$$\mbox{\textbf{M}}(p):=\begin{bmatrix}
   0 & 1 & 0 & \cdots & 0 \\
0 & 0 & 1 &  & 0 \\
\vdots & & & \ddots & \vdots \\
0 & 0 & 0 &  & 1 \\
-p_0 & -p_1 & -p_2 & \cdots &-p_{k-1} 
  \end{bmatrix}.$$
The $\F_q$-algebra $\F_q[P]$ is a finite field with $q^k$ elements.
\end{lemma}

\begin{notation}
Let $V$ be a vector space over a field $\F$ and let $S \subseteq V$ be any
subset. The vector space generated by $S$, i.e.,
the smallest vector subspace of $V$ containing $S$, is denoted by $\langle S
\rangle$.
We always have $\dim_{\F} \langle S \rangle \le |S|$.
\end{notation}

\begin{lemma}\label{lll}
 Let $V$ be a finite-dimensional vector space over a field $\F$. Let $D
\subseteq V$ be any subset
 and set
 $d:=\dim_{\F} \langle D \rangle$. Choose a finite subset $S \subseteq D$. Then
$\dim_{\F}
 \langle D \setminus S \rangle \ge d-|S|$.
\end{lemma}

\begin{proof}
 Since $D=(D\setminus S) \cup S$ we have $\langle D\setminus S \rangle + \langle
S \rangle \supseteq
 \langle (D\setminus S) \cup S \rangle = \langle D \rangle$.
 It follows $$\dim_\F \langle D\setminus S \rangle +\dim_\F \langle S \rangle
\ge
 d+\dim_{\F} \langle D \setminus S \rangle \cap \langle S \rangle.$$
 Since $\dim_\F \langle S \rangle \le |S|$ we conclude $\dim_\F \langle
D\setminus S \rangle +|S| \ge d$.
\end{proof}

\begin{theorem}\label{th}
 Let $q$ be a prime power and let $\F_q$ be the finite field with $q$ elements.
Choose integers
$1 \le k < n$ and write $n=hk+r$ with $0 \le r \le k-1$. Assume $h \ge 2$.
Let $p,p' \in \F_q[x]$ be  two irreducible monic polynomials of 
degree $k$ and $k+r$ respectively, and let $P:=\mbox{\textbf{M}}(p)$,
$P':=\mbox{\textbf{M}}(p')$
be their companion matrices.
For any $1 \le i \le h-1$ set $$\mathcal{M}_i(p,p'):= \left\{ \begin{bmatrix}
  0_k & \cdots & 0_k & I_k & A_{i+1} & \cdots & A_{h-1} & A_{(k)} 
  \end{bmatrix} \ : \ A_{i+1},..., A_{h-1} \in \F_q[P], \ A \in \F_q[P']
\right\},$$ where $0_k$ is the $k \times k$ matrix with zero entries, 
$I_k$ is the $k \times k$ identity matrix, and $A_{(k)}$ denotes the last $k$
rows of $A$.
The set  
$$\mC:= \bigcup_{i=1}^{h-1}\left\{ \mbox{rowsp} (M) : M \in \mathcal{M}_i(p,p')
\right\}
\cup \left\{ \mbox{rowsp} \begin{bmatrix} 0_k & \cdots & 0_k & 0_{k \times r} &
I_k
 \end{bmatrix} \right\}$$
is a partial spread code in $\F_q^n$ of dimension $k$. In particular, the minimum
distance of $\mC$ is $2k$.
\end{theorem}

\begin{proof}
 Choose matrices $M_1 \neq M_2 \in \mathcal{M}_i(p,p')$ and set
$V_1:=\mbox{rowsp}(M_1)$,
 $V_2:=\mbox{rowsp}(M_2)$.  
Since by definition $d(V_1,V_2)=2k-2\dim(V_1 \cap V_2)$, we have
$d(V_1,V_2)=2k$ if and only if
$$\mbox{rk} \begin{bmatrix} M_1 \\ M_2 \end{bmatrix}=2k.$$
Since $M_1 \neq M_2$, it is possible to find
either in
$\begin{bmatrix} M_1 \\ M_2 \end{bmatrix}$, or in $\begin{bmatrix} M_2 \\ M_1
\end{bmatrix}$,
a submatrix in one of the following three forms:
$$N_1:= \begin{bmatrix} I_k  & B \\ 0_k & I_k \end{bmatrix}, \ \ \ \ \
N_2:= \begin{bmatrix} I_k  & B_1 \\ I_k & B_2 \end{bmatrix}, \ \ \ \ \
N_3:= \begin{bmatrix} I_k  & X_{(k)} \\ I_k & Y_{(k)} \end{bmatrix},$$
 with $B, B_1 \neq B_2 \in \F_q[P]$ and $X \neq Y \in \F_q[\tilde{P}]$. Let us
compute the
ranks of such matrices case by case. The rank of $N_1$ is easily computed as
$$\dim_{\F_q} \mbox{rowsp} \begin{bmatrix} I_k & B \end{bmatrix} +
\dim_{\F_q} \mbox{rowsp} \begin{bmatrix} 0_k & I_k \end{bmatrix}- \dim_{\F_q}
\left( \mbox{rowsp} \begin{bmatrix} I_k & B \end{bmatrix} \cap \mbox{rowsp} 
\begin{bmatrix} 0_k & I_k \end{bmatrix} \right)=2k.$$
The rank of $N_2$ is equal to the rank of 
$$\begin{bmatrix} I_k  & B_1 \\ 0_k & B_2-B_1 \end{bmatrix}.$$ 
Since $B_1 \neq B_2$, by Lemma \ref{inver} we get that $B_2-B_1$ is an
invertible matrix, hence
$$\det \begin{bmatrix} I_k  & B_1 \\ 0_k & B_2-B_1 \end{bmatrix} =\det(B_2-B_1)
\neq 0.$$
It follows that $\mbox{rk} (N_2)=2k$.
In order to study the latter case, consider the
$2(k+r) \times 2(k+r)$ matrix $$H:=\begin{bmatrix} I_{k+r} & X \\ I_{k+r} & Y
\end{bmatrix}.$$
By using the same argument 
as above, we get $\mbox{rk}(H)=2(k+r)$. Delete from $H$ the rows from one to $r$
and from
$k+r+1$ to $k+2r$. A matrix of size $2k \times (2k+2r)$, say $\tilde{H}$, is
obtained.
We observe that the rows of $\tilde{H}$ are exactly the rows of
 $N_3$ with $r$ extra zeroes in
the beginning. In particular, 
$\mbox{rk}(\tilde{H})= \mbox{rk} (N_3)$. 
By Lemma \ref{lll} we get $\mbox{rk}(\tilde{H}) \ge 2(k+r)-2r =2k$ and so
$\mbox{rk} (N_3)=2k$.
To conclude the proof, take a matrix $M_1 \in \mathcal{M}_i(p,p')$ and set
$M_2:=\begin{bmatrix} 0_k & \cdots & 0_k & 0_{k \times r} & I_k
 \end{bmatrix}$. It follows $$\mbox{rk} \begin{bmatrix} M_1 \\ M_2
\end{bmatrix}=2k.$$
These arguments prove that that $\mC$ is a set of $k$-dimensional vector
subspaces of $\F_q^n$, whose pairwise intersections are trivial.
\end{proof}

\begin{notation}\label{nota}
The partial spread code $\mC$ defined in the statement of Theorem \ref{th} will
be denoted by 
$\mC_q(k,n;p,p')$.
Since, for any code $\mC \subseteq \mG_q(k,n)$, $\mathcal{C}^\bot \subseteq
\mG_q(n-k,n)$
is a code with the same cardinality and the same distance distribution
as $\mC$ (see \cite{KK1}, Section III), we always assume $1 \le k \le n/2$.
\end{notation}

\begin{remark}
 Partial spread codes provide a generalization of spread codes (see
\cite{GR}, Definition 2).
 Indeed, it is easily seen that spread codes are obtained by taking $r:=0$ and
$p':=p$ in the 
statement of Theorem \ref{th}. On the other hand, partial spread codes exist 
also when $k$ does not divide $n$.
\end{remark}

\begin{example}
Here we construct a partial spread code of lenght $7$ and dimension $2$ over the
binary field $\F_2$. Let $(q,k,n):=(2,2,7)$
and observe that $n \equiv 1 \mod k$. Hence, in the notation of Theorem
\ref{th}, $r=1$. Take irreducible monic polynomials 
$p:=x^2+x+1, p':=x^3+x+1 \in \F_2[x]$ of degree $k$ and $k+r$, respectively. The
companion matrices of $p$ and $p'$ are easily computed as follows:
$$P:= \mbox{\textbf{M}}(p)= \begin{bmatrix} 0 & 1 \\ 1 & 1 \end{bmatrix}, \ \ \
\ 
\ P':= \mbox{\textbf{M}}(p')= \begin{bmatrix} 0 & 1 & 0 \\ 0 & 0 & 1 \\ 1 & 1 &
0
\end{bmatrix}.$$
As a consequence, the elements of $\mC_2(2,7;p,p')$ are the row spaces of all
the matrices in the following forms:
$$\begin{bmatrix} \begin{array}{cc} 1 & 0 \\ 0 & 1 \end{array} & A_1 & A_{(2)}
\end{bmatrix}, \ \ 
\begin{bmatrix} \begin{array}{cccc} 0 & 0 & 1 & 0 \\ 0 & 0 & 0 & 1 \end{array} &
B_{(2)} \end{bmatrix}, \ \ 
\begin{bmatrix} 0 & 0 & 0 & 0 & 0 & 1 & 0 \\ 0 & 0 & 0 & 0 & 0 & 0 & 1
\end{bmatrix},$$
where $A_1$ is any matrix in $\F_q[P]$ and $A_{(2)}, B_{(2)}$ denote the last
two rows
of any $A, B \in \F_q[P']$.
It can be checked that $\mC_2(2,7;p,p')$ has $2^2 \cdot 2^3+2^3+1=41$ elements.
The cardinality
computation will be easily generalized in Proposition
\ref{numb}.
\end{example}

\section{Some properties of partial spread codes}\label{sec3}

In this section we discuss some relevant properties of partial spread codes
introduced in Theorem \ref{th}. In particular,  Proposition
\ref{numb} provides their size and Proposition \ref{max} proves their
maximality,
 with respect to inclusion, as collections of $k$-dimensional vector
subspaces of $\F_q^n$ with trivial pairwise intersections.

\begin{proposition}\label{numb}
 Let $\mC:= \mC_q(k,n;p,p')$ be a partial spread code. The size of $\mC$ is
given by the
formula 
$$|\mC| = \frac{q^n-q^r}{q^k-1}-q^r+1.$$
\end{proposition}
\begin{proof}
We follow the notation of Theorem \ref{th}. 
Let $X,Y$ be matrices in $\F_q[P']$ and 
assume $X_{(k)}=Y_{(k)}$.
If $X \neq Y$ we have
$$\mbox{rk} \begin{bmatrix}
            I_{k+r} & X \\ I_{k+r} & Y
           \end{bmatrix} =2(k+r)
$$ and so, as in the proof of Theorem \ref{th},
$$\mbox{rk} \begin{bmatrix}
            I_{k} & X_{(k)} \\ I_{k} & Y_{(k)}
           \end{bmatrix} =2k
,$$ which yields a contradiction. It follows that $X=Y$. Notice that
the matrices in the statement of Theorem \ref{th} are given in
row-reduced echelon form, which is canonical (see \cite{CM}, Chapter 2.2). As a
consequence, the size of $\mC$ is
easily computed as  
 $$|\mC|= 1+q^{k+r} \sum_{i=0}^{h-2} q^{ki} = (q^n-q^r)/(q^k-1)-q^r+1,$$
as claimed.
\end{proof}

\begin{corollary}
 Let $\mC:=\mC_q(k,n;p,p')$ be a partial spread code. Denote by $\mA_q(k,n,2k)$
the largest possible size
 of a network code in $\mG_q(k,n)$ of minimum distance $2k$. Let $r$ be the
remainder obtained dividing $n$ by $k$. Then
 $$\mA_q(k,n,2k)-|\mC| \le q^r-1.$$
\end{corollary}

\begin{proof}
 Combine Lemma \ref{bound} and Proposition \ref{numb}.
\end{proof}

\begin{remark}
In \cite{EV} T. Etzion and A. Vardy provide a construction of partial spread codes
(see the proof of Theorem 11). Their codes have the same cardinality
and minimum distance as $\mC_q(k,n;p,p')$. The main contribution of
this paper is introducing a block-matrices description of partial
spread codes. Thanks to our constrution, in Section \ref{dec} we are able to provide 
an efficient decoding algorithm for partial spread codes. In the next
proposition, we discuss the maximality of partial spread codes.
\end{remark}

\begin{proposition}\label{max}
 Let $\mN_q(k,n,2k)$ be the set of all the possible network codes $\mC \subseteq
\mG_q(k,n)$
of minimum distance $2k$. Let $\mC:= \mC_q(k,n;p,p')$ be a partial spread code.
Then $\mC$ is a
maximal element of $\mN_q(k,n,2k)$ with respect to inclusion.
\end{proposition}
\begin{proof}
We must prove that there is no partial $k$-spread $\mC'$ in $\F_q^n$ such that
$\mC' \supseteq \mC$ and $|\mC'| > |\mC|$.
 Write $n=hk+r$ with $0 \le r<k$ and $h \ge 2$ (see Notation \ref{nota}). 
Define the partial $k$-spread $$\overline{\mC}:=\mC \setminus \left\{
\mbox{rowsp}
 \begin{bmatrix} 0_k & \cdots & 0_k & 0_{k \times r} & I_k  \end{bmatrix}
\right\}.$$
Assume, by contradiction, that there exists a partial $k$-spread $\mC'$ in
$\F_q^n$ such that
$\mC' \supseteq \overline{\mC}$ and $|\mC'| \ge |\overline{\mC}|+2$.
Set $S:= \bigcup \overline{\mC}  \setminus \{ 0\}$.
By combining Theorem \ref{th} and Proposition \ref{numb} we easily compute
$$|\overline{\mC}|=(q^n-q^r)/(q^k-1)-q^r, \ \ \ \ \ \ \ \ |S|=(q^k-1)\cdot
|\overline{\mC}|=q^n-q^{k+r}.$$
The set $X:= \{ x \in \F_q^n : x_i=0 \mbox{ for any } i=1,...,(h-1)k\}$ is a
vector subspace
of $\F_q^n$ of dimension $k+r$. We clearly have an inclusion 
$X \subseteq \F_q^n \setminus S$. Since $$|\F_q^n \setminus
S|=q^n-(q^n-q^{k+r})=q^{k+r},$$ 
we deduce $X = \F_q^n \setminus S$, $\F_q^n =
X \sqcup S$, with $X$ a $(k+r)$-dimensional vector subspace
of $\F_q^n$. Since $\mC' \supseteq \mC \supseteq \overline{\mC}$,
$|\mC'| \ge \overline{\mC}+2$ and for any $s \in S$ there exists a $V_s \in
\overline{\mC}$
such that $s \in V_s$, we deduce the existence of two $k$-dimensional vector
subspaces $V_1,V_2 \in \mC'$
such that $V_1 \cap V_2 = \{ 0 \}$ and $V_1,V_2 \subseteq X$. Since $X$ is a
vector subspace
of $\F_q^n$ containing $V_1 \cup V_2$ and, by definition,
 $V_1+V_2$ is the smallest vector subspace of $\F_q^n$ containing both $V_1$ and
$V_2$,
we conclude $V_1+V_2 \subseteq X$. It follows $$\dim(V_1)+\dim(V_2)-\dim(V_1
\cap V_2) \le \dim(X)$$ and so
$2k \le k+r$, which is a contradiction.
\end{proof}

\begin{remark}
 Proposition \ref{max} ensures that a partial spread code $\mC_q(k,n;p,p')$
cannot be improved (as a network code
 in $\mG_q(k,n)$ of minimum distance $2k$) by adding new codewords.
\end{remark}

\section{The block structure} \label{sec4}

Here we investigate the block structure of partial spread codes introduced in
the statement of Theorem \ref{th}. This will allow us to produce an
efficient decoding algorithm, which we present in the next section.
The results of this section are a generalization of those contained
in \cite{GR2}.

\begin{lemma}\label{identit}
 Let $\mC:=\mC_q(k,n;p,p')$ be a partial spread code and let $V \in \mC$ be a
codeword, say
 $$V:=\mbox{rowsp} \begin{bmatrix}  S_1 & \cdots & S_{h-1} & S  \end{bmatrix},$$
where the 
 $S_i$'s are $k \times k$ matrices and $S$ is a $k\times (k+r)$ matrix. Let $X
\subseteq \F_q^n$
 be a $t$-dimensional vector subspace given as the row space of a matrix of the
form
 $$\begin{bmatrix}  M_1 & \cdots & M_{h-1} & M \end{bmatrix},$$ where the
$M_i$'s are $k \times k$ matrices and $M$ is a $k\times (k+r)$
matrix\footnote{Notice that
$t \le k$. This assumption is not restrictive from the following point
of view: the decoder can stop collecting incoming vectors as soon as it
receives $k$ inputs (as an alternative, $k$ linearly independent
inputs); then it can attempt to decode the collected data.}.
 If $d(V,X) <k$ then $X$ decodes to $V$. Moreover, for any $1 \le i \le h-1$ the
following two facts are equivalent:
 \begin{enumerate}
  \item[(1)] $S_i=0_k$,
  \item[(2)] $\mbox{rk}(M_i) \le (t-1)/2$.
 \end{enumerate}
\end{lemma}

\begin{proof}
Since the minimum distance of $\mC$ is $2k$ (Theorem \ref{th}) and $d(V,X)<k$,
the space $X$ obviously decodes to $V$.
Let us prove \framebox{$(1) \Rightarrow (2)$}. Without loss of generality,
we assume that $\begin{bmatrix}  S_1 & \cdots & S_{h-1} & S  \end{bmatrix}$ is
in
row-reduced echelon form. Assume that for a fixed index $1 \le i \le h-1$ we 
have $S_i=0_k$. Since $d(V,X)<k$, we have $\dim_{\F_q} (V \cap X) >t/2$. By
definition of $\mC$,
 exactly one of the following cases occurs:
\begin{enumerate}
 \item[(a)] there exists an index $1 \le j\le h-1$ with $j \neq i$ such
 that $S_j=I_k$;
\item[(b)] $S_j=0_k$ for any $1 \le j \le h-1$.
\end{enumerate}
 In the former case, let us consider the matrix $M_{ij}$ defined by
$$M_{ij}:= \begin{bmatrix} 0_k & I_k \\ M_i & M_j \end{bmatrix}.$$
We get $\mbox{rk} (M_{ij}) \le \dim(V+X)= k+t-\dim_{\F_q}(V \cap X) < k+t/2$.
Assume by contradiction that $\mbox{rk}(M_i) > (t-1)/2$. By deleting the last
$k$
coloumns 
of $M_{ij}$ (which are linearly independent and do not lie in the space
generated
by
the first $k$)
we easily deduce the following contradiction:
$$k+(t-1)/2 < \mbox{rk} \begin{bmatrix} 0_k & I_k \\ M_i & M_j \end{bmatrix} <
k+t/2.$$
In the latter case, by definition of $\mC$, we have $V=  \mbox{rowsp}
\begin{bmatrix} 0_k & \cdots & 0_k & 0_{k \times r} & I_k
\end{bmatrix}$. Hence
$$k+(t-1)/2 < \mbox{rk} \begin{bmatrix} 0_k & 0_{k\times r} I_k \\ M_i & M
\end{bmatrix} \le 
\dim(V+X)= k+t-\dim_{\F_q}(V \cap X) < k+t/2,$$
a contradiction. Now we prove \framebox{$(2) \Rightarrow (1)$}. Assume
$\mbox{rk}(M_i) \le (t-1)/2$ for some  index $1 \le i \le h-1$. If $S_i
\neq 0_k$ then, by definition of
$\mC$, $\mbox{rk}(S_i)=k$. Denote by $\pi:\F_q^n \to \F_q^k$ the projection on
the coordinates $ki+1, ki+2,..., k(i+1)$.
Since $\mbox{rowsp}(S_i) = \pi(V)$ and $\mbox{rk}(S_i)=k$, we get that
$\pi_{|V}$ is surjective. Since $\dim_{\F_q}(V)=k$, it follows that $\pi_{|V}$
is also injective. As a consequence,
$$\dim_{\F_q}(V \cap X) = \dim_{\F_q} \pi(V \cap X) \le \dim_{\F_q}(\pi(V) \cap
\pi(X)) \le \dim_{\F_q} \pi(X) =\mbox{rk}(M_i) \le (t-1)/2,$$
which contradicts the assumption that $d(V,X)<k$.
\end{proof}

\begin{remark}\label{riassuntino}
 Lemma \ref{identit} has the following useful interpretation. Assume that a
partial spread code $\mC:=\mC_q(k,n;p,p')$ is used for random network coding and
 a $t$-dimensional vector space $X:=\mbox{rowsp} \begin{bmatrix}  M_1 & \cdots &
M_{h-1} & M \end{bmatrix}$
 is received. Assume the existence of a (unique) codeword $V \in \mC$ such that
$d(V,X)<k$
(i.e., $X$ decodes to $V$). If $\mbox{rk}(M_i) \le (t-1)/2$ for any $1 \le i \le
h-1$
 then $V=  \mbox{rowsp} \begin{bmatrix} 0_k & \cdots & 0_k & 0_{k \times r} &
I_k
\end{bmatrix}$. Otherwise, let $i$ denote the smallest integer $1 \le i \le h-1$
such that
 $\mbox{rk}(M_i)>(t-1)/2$. Then there exist unique matrices $A_{i+1},...,A_{h-1}
\in \F_q[P]$ and a unique matrix $A \in \F_q[P']$ such that
 $V= \mbox{rowsp}\begin{bmatrix} 0_k & \cdots & 0_k & I_k & A_{i+1} & \cdots &
A_{h-1} &
A_{(k)} \end{bmatrix}$, where the identity matrix $I_k$ is the $i$-th
$k \times k$ block. 
\end{remark}

\begin{lemma}\label{red}
With the setup of Remark \ref{riassuntino}, assume that $V \neq  \mbox{rowsp} 
\begin{bmatrix} 0_k & \cdots & 0_k & 0_{k \times r} & I_k
\end{bmatrix}$. For any $i+1 \le j \le h-1$ we have
$$d\left( \mbox{rowsp} \begin{bmatrix}  I_k & A_j \end{bmatrix} , \mbox{ rowsp}
\begin{bmatrix} M_i & M_j
 \end{bmatrix} \right) < k, \ \ \ \ \ d\left( \mbox{rowsp} \begin{bmatrix} I_k &
A_{(k)}
 \end{bmatrix} , \mbox{ rowsp} \begin{bmatrix} M_i & M \end{bmatrix} \right) <
k.$$
\end{lemma}
\begin{proof}
 Fix an integer $j$ such that $i+1 \le j \le h-1$ and denote by $\pi:\F_q^n \to
\F_q^{2k}$ the
projection
on the coordinates $ki+1, ki+2,...,k(i+1), kj+1,kj+2,...,k(j+1)$. We have
$\pi(V)=
\mbox{rowsp} \begin{bmatrix} I_k & A_j\end{bmatrix}$ and $\pi(X)=
\mbox{rowsp} \begin{bmatrix} M_i & M_j\end{bmatrix}$. In particular, $\mbox{rk}
(\pi_{|V})=k$. As a consequence,
$\dim_{\F_q} \ker (\pi_{|V}) \le k-k=0$ and so $\pi_{|V}$ is injective.
By the trivial inclusion of vector spaces $\pi(V \cap X) \subseteq \pi(V) \cap
\pi(X)$ it follows
$\dim_{\F_q} \pi(V \cap X) \le \dim_{\F_q} (\pi(V) \cap \pi(X))$. Hence
\begin{eqnarray*}
 d(\pi(V), \pi(X)) &=& k+\dim_{\F_q} \pi(X)-2\dim_{\F_q} (\pi(V) \cap \pi(X)) \\
&\le& k+t - 2\dim_{\F_q} \pi(V \cap X) \\
&=& k+t - 2\dim_{\F_q} (V \cap X) \\
&=& d(V,X) \\
&<& k.
\end{eqnarray*}
In order to prove that $d\left( \mbox{rowsp} \begin{bmatrix} I_k & A_{(k)}
 \end{bmatrix} , \mbox{ rowsp} \begin{bmatrix} M_i & M \end{bmatrix} \right) <
k$
we may notice
that the same argument still works if we choose as $\pi:\F_q^n \to \F_q^{2k+r}$
the projection
 on the coordinates
$ki+1, ki+2,...,k(i+1), k(h-1)+1,k(h-1)+2,...,kh, kh+1, ..., kh+r$. 
\end{proof}

\begin{remark}\label{2casi}
 By Lemma \ref{red}, when decoding a partial spread code we may restrict
to one of the two the cases $n=2k$ and $n=2k+r$, with $1 \le r \le
k-1$. Moreover, the lemma allows us to parallelize 
the computation, reducing the decoding complexity to the case $n=2k+r$.
\end{remark}

\section{Decoding partial spread codes} \label{dec}

In  \cite{KK1} R. K\"{o}tter and F. R. Kschischang illustrate a general
network code construction and a related efficient algorithm to decode them.
A more efficient algorithm to decode the same codes appears in \cite{KK2}.
After recalling the definition of \textit{Reed-Solomon like} code,
we use the results established in the previous section to adapt
any decoding algorithm for such codes to partial spread codes of the
form $\mC_q(k,n;p,p')$.

\begin{definition}
Let $q$ be a prime power and let $n>1$ be an integer. Let 
$A:=\{ \alpha_1,...,\alpha_k \} \subseteq \F_{q^n}$
be a set of $\F_q$-linearly independent elements.
Choose an integer $s \le k$ and denote by $\F_{q^n}^s[x]$ the vector
space of the linearized polynomial of degree at most $s$ and coefficients in
$\F_q$ (see \cite{KK1}, Section 5.A, for details).
 Fix an $\F_q$-isomorphism of vector spaces $\varphi: \F_{q^n} \to \F_q^n$.
The \textbf{Reed-Solomon like} code associated to the $6$-tuple $(q, n, k, s, A,
\varphi)$ is
the set
$$\mbox{\textbf{KK}}_q(n, k, s, A, \varphi):= \left\{ \mbox{rowsp}
\begin{bmatrix}
 &  &  &  &   & \varphi(f(\alpha_1)) \\
 &  &  & & &  \varphi(f(\alpha_1)) \\ 
 &  & I_k & & & \vdots \\
& &  &  & &  \varphi(f(\alpha_{k-1})) \\
& &  &   &  & \varphi(f(\alpha_k))   \end{bmatrix} \ : \ f \in \F_{q^n}^s[x]
\right\}.$$
\end{definition}

\begin{remark}
 A Reed-Solomon like code $\mbox{\textbf{KK}}_q(n, k, s, A, \varphi)$ is
a subset of the Grassmannian variety $\mG_q(k,k+n)$. As a consequence, it is a
$q$-ary 
network code of lenght $k+n$ and dimension $k$. The size of such a code
is given by the easy-computable formula
$|\mbox{\textbf{KK}}_q(n, k, s, A, \varphi)|=q^{sn}$. See \cite{KK1}, Section
5.1,
for a more detailed discussion.
\end{remark}

\begin{lemma}\label{comp}
Let $q$ be a prime power and let $k \ge 1$ be an integer.  Let
$p$ an irreducible monic polynomial $p \in \F_q[x]$ of degree $k$ and let
$P:=\mbox{\textbf{M}}(p)$
be its companion matrix. Choose a root $\lambda \in \F_{q^k}$ of $p$.
 Denote by $\varphi:\F_{q^n} \to \F_q^n$ the $\F_q$-isomorphism defined, for any
$0 \le i \le k-1$,
 by $\lambda^i \mapsto e_{i+1}$, where $\{ e_1,...,e_k\}$ is the canonical basis
of $\F_q^k$.
Let $A \in \F_q[P]$ and, for any $1 \le i \le k$, let $A^i \in \F_q^k$ denote
the $i$-th row 
of $A$. For any $1 \le i \le k$ we have $\varphi^{-1}(A^i)=\lambda^{j-1}
\varphi^{-1}(A^1)$.
In particular, if $f \in \F_{q^k}^1[x]$ is defined by $f(x):=\varphi^{-1}(A^1)
x$,
then $$A= \begin{bmatrix}  \varphi(f(\lambda^0)) \\  \varphi(f(\lambda)) \\
 \varphi(f(\lambda^2)) \\ \vdots \\ \varphi(f(\lambda^{k-1})) \end{bmatrix}.$$
\end{lemma}
\begin{proof}
 Use \cite{GR2}, Proposition 15, with $n=k$.
\end{proof}

\begin{notation}
In the costruction of a partial spread code $\mC_q(k,2k+r,p,p')$ with $0 \le r
\le k-1$, 
the companion matrix of $p$ is never involved (see Theorem \ref{th}). As a
consequence, we write
$\mC_q(k,2k+r;p')$ in this case.
\end{notation}

\begin{remark}
 By Lemma \ref{2casi}, in order to decode a partial spread code
$\mC_q(k,n;p,p')$ we may
restrict to decoding partial spread codes of the form $\mC_q(k,2k+r;p)$, with $0
\le r \le k-1$.
The case $r=0$ is easily solved. Indeed, by Lemma \ref{comp},
$\mC_q(k,2k;p) \setminus \left\{ \mbox{rowsp} \begin{bmatrix} 0_k & I_k  
\end{bmatrix}
 \right\}$ is a Reed-Solomon like code and so we may simply proceed as in the
following  Algorithm
 \ref{A1}.
\end{remark}

\begin{algorithm}\label{A1}
\caption{Decoding a $\mC_q(k,2k;p)$ code.}
\begin{algorithmic}

 %\KwData{
 
 \STATE \textbf{Data}: a decodable\footnote{A vector space $X$ is said to be
\textbf{decodable}
with respect to a network code $\mC$ if there exists a codeword $V \in \mC$ such
that
$d(V,X) \le \lfloor (d(\mC)-1)/2\rfloor$, $d(\mC)$ being the minimum distance of
$\mC$. Such a 
codeword is clearly unique.}
 $t$-dimensional row space, $X$, of a 
$(k \times 2k)$-matrix $\begin{bmatrix} M_1 & M_2\end{bmatrix}$.

\STATE  \textbf{Result}: the unique $V \in \mC_q(k,2k;p)$ such 
that $d(V,X)<k$, given as a matrix in row-reduced echelon form whose row space
is $V$.

 \IF{$\mbox{rk}(M_1) \le (t-1)/2$}
 
 \STATE $V=\mbox{rowsp} \begin{bmatrix} 0_k & I_k
\end{bmatrix}$.

\ELSE 

\STATE use a decoding algorithm for Reed-Solomon like codes on $\mC_q(k,2k;p)
\setminus
\left\{ \mbox{rowsp} \begin{bmatrix} 0_k & I_k   \end{bmatrix}
 \right\}$.
 
\ENDIF

\end{algorithmic}
\end{algorithm}

\begin{remark}
In \cite{GR2}, a decoding procedure for $\mC_q(k,hk;p)$ spread codes
which is independent of those of \cite{KK1} and \cite{KK2} is
proposed. Lemma~\ref{red} allows us to apply the decoding algorithm
from~\cite{GR2} to partial spread codes. This algorithm is particularly
efficient in the case $k\ll n$.
 \end{remark}

Now we focus on a decoding procedure for partial spread codes of the form
$\mC_q(k,2k+r;p)$ with $1 \le r \le k-1$. To be precise, in the following
Proposition \ref{emb}
we construct a canonical \textit{embedding} of a partial spread code
$\mC_q(k,2k+r;p)$
into the spread code $\mC_q(k+r,2(k+r);p)$. Any decoding procedure for
$\mC_q(k+r,2(k+r);p)$
gives, in this way, a decoding procedure for $\mC_q(k,2k+r;p)$.

\begin{proposition}\label{emb}
 Let $\mC:=\mC_q(k,2k+r,p)$ be a partial spread code with $1 \le r \le k-1$.
Let 
$X:= \mbox{rowsp} \begin{bmatrix} M_1 & M \end{bmatrix}$ be a
$t$-dimensional
vector space in $\F_q^{2n+r}$, where $M_1$ is a $(k \times k)$-matrix and $M$ is
a matrix of
size $k \times (k+r)$. Assume the existence of a matrix $A \in \F_q[P]$ such
that
$d \left( \mbox{rowsp} \begin{bmatrix} I_k & A_{(k)} \end{bmatrix} , 
\mbox{rowsp}  \begin{bmatrix} M_1 & M \end{bmatrix}\right) <k$. Define the
following two
$(k+r)\times (k+r)$-matrices:
$$\overline{M}_1:= \begin{bmatrix} 0_r & 0_{r \times k} \\ 0_{k \times r} & M_1
\end{bmatrix},
\ \ \ \ \ \overline{M}:= \begin{bmatrix} 0_{r\times (k+r)} \\ M \end{bmatrix}.$$
We have  $$d \left( \mbox{rowsp} \begin{bmatrix} I_{k+r} & A \end{bmatrix} , 
\mbox{ rowsp}  \begin{bmatrix} \overline{M}_1 & \overline{M}
\end{bmatrix}\right)
<k+r.$$
\end{proposition}

\begin{proof}
Set $V:= \mbox{rowsp} \begin{bmatrix} I_k & A_{(k)} \end{bmatrix}$ and observe
that the hypothesis
$d(V,X)<k$ can be restated as $\dim(V \cap X) > t/2$. Define $\overline{V}:=
\mbox{rowsp}
 \begin{bmatrix} I_{k+r} & A \end{bmatrix}$
and $\overline{X}:= \mbox{rowsp} \begin{bmatrix} \overline{M}_1 & \overline{M}
\end{bmatrix}$.
By construction, $\dim_{\F_q} X= \dim_{\F_q} \overline{X}=t$ and
$\dim_{\F_q} (\overline{V} \cap \overline{X})  \ge \dim_{\F_q} (V \cap X)$. It
follows
\begin{eqnarray*}
 d(\overline{V}, \overline{X}) &=& \dim_{\F_q} \overline{V}+ \dim_{\F_q}
\overline{X}-
2\dim_{\F_q} (\overline{V} \cap \overline{X}) \\
&=& k+r+t-2\dim_{\F_q} (\overline{V} \cap \overline{X}) \\
&\le& k+r+t-2\dim_{\F_q} (V \cap X) \\
&<& k+r+t-2(t/2) \\
&=& k+r,
\end{eqnarray*}
as claimed.
\end{proof}

\begin{remark}
 Proposition \ref{emb} has the following useful interpretation. Assume that a
partial spread
code $\mC_q(k,2k+r;p)$ is given, with $1 \le r \le k-1$, and
 $X:= \mbox{rowsp} \begin{bmatrix} 
M_1 & M \end{bmatrix}$ is received ($M_1$ and $M$ being as in the statement of
the proposition).
Then we may construct the matrices $\overline{M}_1$ and $\overline{M}$ as
described and
consider the vector space  
$\overline{X}:= \mbox{rowsp} \begin{bmatrix} \overline{M}_1 & \overline{M}
\end{bmatrix}$.
The minimum distance of the  (partial) spread code $\mC_q(k+r, 2(k+r);p)$ is
$2(k+r)$. 
By Proposition \ref{emb}, if $X$ decodes to 
$V:=\mbox{rowsp} \begin{bmatrix} I_k & A_{(k)}\end{bmatrix}$ in
$\mC_q(k,2k+r;p)$, then
$\overline{X}$ decodes to $\overline{V}:=\mbox{rowsp} \begin{bmatrix} I_{k+r} &
A \end{bmatrix}$
in $\mC_q(k+r,2(k+r);p)$.
It follows that Algorithm \ref{A1} (with $k \leftarrow k+r$) applied to
$\overline{X}$ produces 
$\begin{bmatrix} I_{k+r} & A \end{bmatrix}$. Finally, $V$ is the rowspace of the
matrix 
obtined by deleting the first $r$ rows and the first $r$ coloumns of
$\begin{bmatrix} I_{k+r} & A \end{bmatrix}$. This discussion leads to the
following
 Algorithm \ref{A2}.
\end{remark}

\begin{algorithm}\label{A2}
\caption{Decoding a $\mC_q(k,2k+r;p)$ code with $1 \le r \le k-1$.}
\begin{algorithmic}
 \STATE \textbf{Data}: a decodable $t$-dimensional row space, $X$, of a 
$(k \times 2k+r)$-matrix $\begin{bmatrix} M_1 & M \end{bmatrix}$.
 \STATE \textbf{Result}: the unique $V \in \mC_q(k,2k+r;p)$ such 
that $d(V,X)<k$, given as a matrix in row-reduced echelon form whose row space
is $V$.

 \IF{$\mbox{rk}(M_1) \le (t-1)/2$}
 
 \STATE $V=\mbox{rowsp} \begin{bmatrix} 0_k & 0_{k
\times r} & I_k 
\end{bmatrix}$.

\ELSE

\STATE construct the matrix $\begin{bmatrix} \overline{M}_1 &
\overline{M}\end{bmatrix}$ as
explained in Lemma \ref{emb}. Then use Algorithm \ref{A1} with $\mC_q(k+r,
2(k+r);p)$
on 
$\begin{bmatrix} \overline{M}_1 & \overline{M}\end{bmatrix}$. Delete the first
$r$ rows and the first $r$ coloumns of the output.

\ENDIF
\end{algorithmic}
\end{algorithm}

\begin{remark}
By Proposition \ref{emb}, in Algorithm \ref{A2} we may replace the use of
Algorithm \ref{A1} with any other decoding algorithm for spread codes. 
\end{remark}

\section*{Conclusions}
In this paper we provide an easy description of partial spreads over
finite fields, whose interest dates back to classical problems in
projective geometry. We suggest the use of partial spreads as network codes,
investigating the mathematical properties due to our construction, proving
their maximality, and providing a decoding algorithm for them.

\section*{Acknowledgment}
The authors would like to thank Leo Storme for useful discussions on
partial spreads in finite projective geometry.

\end{document}